%% file: main.tex
\documentclass[twocolumn]{autart}   
\usepackage{amsmath,amssymb,amsfonts}
\usepackage{graphicx}
\usepackage{color}
\usepackage{bbding}
\usepackage{wrapfig}
\usepackage{pifont}
\usepackage{cases}
\usepackage{subfigure}
\usepackage{balance}
\usepackage[T1]{fontenc}
\usepackage{cite}
\usepackage{natbib}
\usepackage[justification=centering]{caption}
\usepackage{graphicx}
\usepackage{algorithm}
\usepackage{algorithmic}
\usepackage{mathtools}

\usepackage{bm}
\usepackage{listings}

\newtheorem{theorem}{Theorem}
\allowdisplaybreaks
\newtheorem{proposition}{Proposition}
\newtheorem{lemma}{Lemma}
\newtheorem{example}{Example}

\newtheorem{assumption}{Assumption}
\newtheorem{problem}{Problem}
\newcommand{\oomit}[1]{}

\begin{document}

\begin{frontmatter}

\title{Sufficient and Necessary Continuous Barrier-like Conditions for Discrete-Time Stochastic Reach-Avoid Verification} 

\thanks[footnoteinfo]{This work is funded by the National Research Foundation, Singapore, under its RSS Scheme (NRF-RSS2022-009) and the Basic Research Program of  Institute of Software, CAS (Grant No. ISCAS-JCMS-202302). Corresponding author Bai Xue. Tel. +8613552483249.}

\author[]{Bai Xue}\ead{xuebai@ios.ac.cn}  

\address{KLSS, Institute of Software, Chinese Academy of Sciences, Beijing, China}


\begin{abstract}                          
This paper develops necessary and sufficient barrier-like characterizations using continuous barrier functions for infinite-horizon reach-avoid verification of discrete-time stochastic systems. Existing results establish necessary and sufficient conditions in terms of functional inequalities involving measurable or lower semicontinuous barrier functions. However, the limited regularity of such functions may hinder their numerical approximation and computational synthesis. Building on our previous barrier-like condition for finite-horizon reach-avoid verification, we show that this condition can also be used for infinite-horizon reach-avoid verification and, under a uniform absolute continuity condition on the transition kernels, admits a continuous barrier function whenever the exact reach-avoid probability is strictly larger than the prescribed threshold for every state in the initial set. We further show that the resulting continuous barrier function can be uniformly approximated by a polynomial one while preserving the required barrier-like conditions. For polynomial systems, we formulate these conditions as polynomial positivity constraints over compact basic semialgebraic sets. Putinar's Positivstellensatz then converts the positivity conditions into sum-of-squares (SOS) certificates, yielding semidefinite programming (SDP) formulations for synthesizing polynomial barrier functions. We establish both soundness and completeness of the resulting SOS-based procedure. Finally, two numerical examples illustrate the theoretical results and demonstrate the resulting SDP approach.
\end{abstract}

\end{frontmatter}

\input{introduction}

\input{pre}

\input{conditions}
\input{sos}

\input{examples}

\input{conclusion}
\input{ack}
\bibliographystyle{unsrt}
\bibliography{ref}

\end{document}

%% file: introduction.tex
\section{Introduction}
\label{sec:introduction}
Reach-avoid verification is a fundamental problem in the analysis of autonomous and learning-enabled systems operating under uncertainty \cite{sullivan2015introduction,thiebes2021trustworthy}. In
many applications, including robotics, autonomous vehicles, and
cyber-physical systems, system dynamics are influenced by stochastic disturbances, making the behavior of the system inherently probabilistic. This has motivated the development of probabilistic verification methods that provide formal guarantees on the probability of satisfying reach-avoid specifications \cite{baier2008principles,franzle2008stochastic,abate2008probabilistic,abate2010approximate}. 

Given a collection of sets, including a safe set, an initial set, and a target set, reach-avoid verification for stochastic systems concerns whether the probability that the system, starting from the initial set, reaches the target set while remaining within the safe set until reaching the target is at least a prescribed threshold. We refer to this probability as the \emph{reach-avoid probability}. Depending on whether the target-hitting time is bounded or unbounded, reach-avoid verification is classified as finite-horizon or infinite-horizon verification. In this work, we focus on infinite-horizon reach-avoid verification.

Among existing methods, barrier-function-based methods are particularly attractive for infinite-horizon reach-avoid verification because they formulate the verification problem as the existence of barrier functions satisfying barrier-like conditions expressed as functional inequalities, without explicitly enumerating system trajectories. These functional inequalities are often tractable to verify or synthesize computationally; see, e.g., \cite{prajna2007framework,prajna2007convex}. However, most barrier-function-based methods provide only sufficient barrier-like conditions for infinite-horizon reach-avoid verification, resulting in an incomplete characterization. That is, even when a system satisfies the reach-avoid specification, the existence of a barrier function satisfying the corresponding barrier-like conditions is not guaranteed. Consequently, a purely sufficient characterization cannot, in general, distinguish between a genuine violation of the specification and a failure to find a suitable barrier function, since the nonexistence of a barrier function does not imply that the specification is violated. This limitation has motivated recent efforts toward converse characterizations. Recently, \cite{xue2026sufficient} proposed a barrier-like condition for infinite-horizon reach-avoid verification of discrete-time stochastic systems by relaxing the Bellman equation. It was further shown that, when the exact reach-avoid probability is strictly larger than the prescribed threshold and the initial set is a singleton, this sufficient condition is also necessary. This converse result was subsequently extended in \cite{xue2026converse} to discrete-time stochastic systems with a compact initial set. That is, whenever the exact reach-avoid probability is uniformly strictly larger than the prescribed threshold over the initial set, there exists a lower semicontinuous (l.s.c.) barrier function satisfying the corresponding barrier-like condition. This converse result establishes that, at the level of general l.s.c. barrier functions, the barrier formulation itself need not introduce conservatism. 

However, an l.s.c. function may exhibit sharp spatial variations, discontinuities, and nondifferentiable features, which complicate numerical approximation and computational synthesis. In particular, a discontinuous l.s.c. function cannot, in general, be uniformly approximated arbitrarily closely by continuous functions. Consequently, directly replacing an l.s.c. barrier function by a continuous finite-dimensional parameterization does not automatically preserve the barrier-like inequalities or their associated guarantees. This motivates the study of continuous barrier functions, whose regularity is more amenable to standard optimization and learning-based synthesis methods.

Obtaining such a continuous strengthening is nontrivial because the barrier-like conditions couple the barrier function with the system's stochastic transition operator. Moreover, the desired regularization cannot rely on uniform approximation of the original l.s.c. barrier function, since such approximation may be impossible. Instead, we resort to a different barrier-like condition, established in Theorem 4 of \cite{xue2024finite} for finite-horizon reach-avoid verification. We first show that this condition also provides a sufficient condition for infinite-horizon reach-avoid verification. We then show that, under a uniform absolute continuity condition on the transition kernels, whenever the exact reach-avoid probability is uniformly strictly larger than the prescribed threshold over the initial set, there exists a continuous barrier function satisfying this condition. The continuous barrier function is constructed from the l.s.c. barrier function established in \cite{xue2026converse}. The resulting continuous converse characterization therefore provides a genuine regularity improvement over \cite{xue2026converse}. The continuous characterization further provides a direct route to polynomial barrier functions. On a compact computational domain, continuous functions can be uniformly approximated by polynomials. Thus, a continuous barrier function can be approximated by a polynomial while preserving the barrier-like inequalities and the reach-avoid guarantee. When the system has the polynomial structure, these conditions can be formulated as polynomial positivity constraints and subsequently handled using SOS programs. This establishes a direct connection between the analytical converse characterization and computational function synthesis. Finally, we demonstrate the effectiveness of the proposed approach through two numerical examples using SDP tools.

The main contributions are summarized as follows.
\begin{enumerate}
\item We establish necessary and sufficient barrier-like conditions for infinite-horizon reach-avoid verification of discrete-time stochastic systems using \emph{continuous} barrier functions under appropriate assumptions.
\item We show that the continuous barrier functions in the converse characterization can be approximated by polynomial barrier functions while preserving the barrier-like inequalities and the resulting reach-avoid guarantee.
\item We develop an SOS-based computational framework for synthesizing polynomial barrier functions for polynomial systems over compact basic semialgebraic regions, and establish the soundness and completeness of the resulting SOS hierarchy.
\end{enumerate}

\section*{Related Work}

Formal verification of stochastic systems generally provides either qualitative guarantees, such as almost-sure satisfaction of specifications \cite{majumdar2024necessary,abate2024stochastic,kordabad2026certificates}, or quantitative guarantees, which seek probabilistic bounds on the satisfaction or violation of such specifications \cite{kushner1967stochastic,prajna2007framework,chakarov2013probabilistic,wang2021safety,takisaka2021ranking,yu2023safe,abate2025quantitative,henzinger2025supermartingale,xue2026quantitative,xue2026new,chen2026construction}. In this setting, barrier-function-based methods provide a particularly useful framework for obtaining tractable verification guarantees.

Barrier-function-based methods were initially proposed for deterministic systems as a formal approach to safety verification \cite{prajna2004safety}. Subsequent work extended and refined these methods to address a broader range of temporal verification and control problems \cite{prajna2007convex,kong2013exponential,anand2022k,ames2019control}. Since practical systems are often subject to stochastic disturbances, barrier-function-based methods have also been extended to stochastic systems. In particular, they have been developed for finite- and infinite-horizon safety verification of both continuous- and discrete-time stochastic systems through the use of Ville's inequality \cite{ville1939etude} (see, e.g., \cite{prajna2007framework,steinhardt2012finite,feng2020unbounded,santoyo2021barrier,anand2022k,mathiesen2022safety,zhi2024unifying,laurenti2025unifying}) and relaxations of associated equations (see, e.g., \cite{yu2023safe,xue2026sufficient,chen2026construction}). These developments have also enabled controller synthesis through control barrier functions \cite{jagtap2020formal,sarkar2020high,wang2021safety,kordabad2024control}.

Beyond safety verification, which primarily concerns avoiding unsafe sets, reach-avoid verification considers the probability of eventually reaching a target set while satisfying prescribed safety constraints. By explicitly accounting for both safety and task completion, reach-avoid verification provides a natural framework for analyzing the reliability of stochastic autonomous and learning-enabled systems and has therefore attracted increasing attention in recent years. For discrete-time stochastic systems, a barrier-like condition was introduced for computing $p$-reach-avoid sets in \cite{xue2021reach}. This condition is derived by relaxing a system of equations whose solution characterizes the exact reach-avoid probability, and can be straightforwardly extended to infinite-horizon reach-avoid verification \cite{xue2026sufficient,cao2025comparative}. This framework was subsequently extended to continuous-time stochastic systems in \cite{xue2024}. Later, reach-avoid supermartingales were introduced in \cite{vzikelic2023learning,vzikelic2023compositional} to certify reach-avoid specifications and facilitate controller synthesis. A comparison of these barrier-like conditions for infinite-horizon reach-avoid verification of discrete-time stochastic systems was presented in \cite{cao2025comparative}. For continuous-time stochastic systems, stochastic Lyapunov-barrier functions were proposed in \cite{meng2022sufficient} to establish sufficient conditions for probabilistic reach-avoid-stay specifications. For finite-horizon specifications, barrier-like conditions were developed in \cite{xue2024finite} and \cite{xue2026new,xue2025refined} for discrete-time and continuous-time systems, respectively.

Despite these advances, converse results for stochastic systems---namely, results establishing that barrier-like conditions are also necessary---remain comparatively underexplored. This contrasts with the more extensive development of converse results for deterministic continuous-time systems modeled by ordinary differential equations; see, e.g., \cite{lin1996smooth,prajna2005necessity,prajna2007convex,wisniewski2015converse,ratschan2018converse,liu2021converse,maghenem2022converse,meng2022smooth,li2026converse}. For discrete-time stochastic systems, necessary and sufficient conditions for almost-sure reachability were established in \cite{majumdar2024necessary}. In probabilistic program verification, termination can be reduced to an almost-sure reachability problem for terminal states, and sound and complete proof rules for both qualitative (i.e., almost-sure) and quantitative termination were recently developed in \cite{majumdar2025sound}. In addition, a complete characterization of reachability certificates for discrete-time stochastic linear systems was provided in \cite{kordabad2026certificates}. For infinite-horizon safety and reach-avoid verification of discrete-time stochastic systems, necessary and sufficient barrier-like conditions were established in \cite{xue2026sufficient} via Bellman relaxations. However, the converse characterization for reach-avoid verification was restricted to a singleton initial set. This limitation was addressed in \cite{xue2026converse}, where a set-based converse characterization was established for infinite-horizon reach-avoid verification. Specifically, whenever the reach-avoid probability is uniformly strictly larger than the prescribed threshold over a compact initial set, there exists an l.s.c. barrier function satisfying the corresponding barrier-like conditions. Although lower semicontinuity is sufficient for the theoretical characterization, it is less amenable to numerical synthesis and finite-dimensional approximation. 

\noindent{\textbf{From l.s.c. to Continuous Functions.}} In this work, we bridge this gap for infinite-horizon reach-avoid verification of discrete-time stochastic systems. Building on the l.s.c. barrier function established in \cite{xue2026converse} and the barrier-like condition introduced in \cite{xue2024finite}, we show that, under a uniform absolute continuity condition on the transition kernel, there exists a continuous barrier function satisfying the latter condition. We then invoke the density of polynomials on compact sets to obtain polynomial barrier functions, thereby establishing a direct connection between the analytical converse characterization and SOS-based computational synthesis for polynomial systems. Thus, the key advance over \cite{xue2026converse} is not the existence of a barrier function itself, but the strengthening of its regularity from lower semicontinuity to continuity and, subsequently, to a computationally tractable polynomial form.

The remainder of the paper is organized as follows. Section \ref{sec:pre} introduces the discrete-time stochastic system, the infinite-horizon reach-avoid specification, and the required assumptions. Section \ref{sec:cbc} develops the continuous barrier-like characterization, while Section \ref{sec:sos} presents the SOS-based computational framework for polynomial barrier functions. Section \ref{sec:ex} provides numerical examples demonstrating the effectiveness of the proposed approach. Finally, Section \ref{sec:con} concludes the paper.

Throughout this paper, we use the following basic notions: $\mathbb{R}$ denotes the set of real values; $\mathbb{N}$ denotes the set of nonnegative integers; for sets $\Delta_1$ and $\Delta_2$, $\overline{\Delta_1}$ denotes the closure of the set $\Delta_1$, and $\Delta_1\setminus \Delta_2$ denotes the difference of sets $\Delta_1$ and $\Delta_2$, which is the set of all elements in $\Delta_1$ that are not in $\Delta_2$. For a set $\Delta \subseteq\mathbb R^n$, we denote by $C(\Delta)$
the space of real-valued continuous functions on $\Delta$, and by
$C^1(\Delta)$ the space of continuously differentiable functions on $\Delta$; $\mathbb{R}[\bm{x}]$denote the ring of real-valued polynomials in the state variables $\bm{x}$.

%% file: pre.tex
\section{Preliminaries}
\label{sec:pre}

This section introduces the stochastic discrete-time systems considered in this paper and formulates the infinite-horizon reach-avoid verification problem. We first define the system dynamics, the reach-avoid hitting time, and the reach-avoid probability. We then formulate the infinite-horizon reach-avoid verification problem and state the assumptions used throughout the paper, including a uniform absolute continuity condition on the transition kernels that controls the effect of exceptional sets arising in the continuous approximation of l.s.c. functions. Finally, we recall the necessary and sufficient barrier-like conditions involving l.s.c. barrier functions established in \cite{xue2026converse}. These conditions serve as the starting point for developing the corresponding necessary and sufficient barrier-like conditions involving continuous barrier functions in the next section.

\subsection{Discrete-Time Stochastic Systems}
\label{sub:ss}

Consider the discrete-time stochastic system
\begin{equation}
\label{systems}
\bm{x}_{k+1}
=
\bm{f}(\bm{x}_k,\bm{\theta}_k),
\qquad k\in\mathbb{N},
\end{equation}
where $\bm{x}_k\in\mathbb{R}^n$ is the state and
$\bm{\theta}_k\in\Theta\subseteq\mathbb{R}^m$ is an i.i.d. disturbance with distribution $\mathbb{P}_{\bm{\theta}}$. We denote the corresponding expectation by $\mathbb{E}_{\bm{\theta}}[\cdot]$, and equip the disturbance space $\Theta$ with its standard Borel $\sigma$-algebra.

Let $\pi=(\bm{\theta}_0,\bm{\theta}_1,\ldots)\in\Theta^\infty$ denote a disturbance realization. For an initial state
$\bm{x}\in\mathbb{R}^n$ and a realization $\pi\in\Theta^\infty$, let $\bm{\phi}_{\bm{x}}^\pi:\mathbb{N}\to\mathbb{R}^n$ denote the corresponding trajectory, with $\bm{\phi}_{\bm{x}}^\pi(0)=\bm{x}$.
Let $\mathcal{X}\subseteq\mathbb{R}^n$ be the safe set,
$\mathcal{X}_r\subseteq\mathcal{X}$ the target set, and
$\mathcal{X}_0\subseteq\mathcal{X}$ the initial set. For $\bm{x}\in\mathcal{X}_0$, the reach-avoid hitting time is defined by
\begin{equation}
\label{hitting}
\tau(\bm{x},\pi):=
\inf\left\{
k\in\mathbb{N} \middle|
\begin{split}
&\bm{\phi}_{\bm{x}}^\pi(k)\in\mathcal{X}_r
 ~\text{and} \\
&\bm{\phi}_{\bm{x}}^\pi(j)\in\mathcal{X}
\ \forall j\leq k
\end{split}
\right\},
\end{equation}
with the convention $\tau(\bm{x},\pi)=\infty$ if the target is never reached while the trajectory remains in the safe set.

The reach-avoid probability from $\bm{x}\in\mathcal{X}_0$ is
\begin{equation}
\label{ra_p}
\mathbb{P}_{\mathrm{RA}}(\bm{x})
:=
\mathbb{P}_{\pi}\bigl(\tau(\bm{x},\pi)<\infty\bigr),
\end{equation}
where $\mathbb{P}_{\pi}$ is the canonical product measure induced by $\mathbb{P}_{\bm{\theta}}$, and $\mathbb{E}_{\pi}[\cdot]$ denotes the corresponding expectation. For convenience, we extend $\mathbb{P}_{\mathrm{RA}}$ to all of $\mathbb{R}^n$ by setting
\[
\mathbb{P}_{\mathrm{RA}}(\bm{x})=0
\quad\text{for }\bm{x}\in\mathbb{R}^n\setminus\mathcal{X},
\quad
\mathbb{P}_{\mathrm{RA}}(\bm{x})=1
\quad\text{for }\bm{x}\in\mathcal{X}_r.
\]

The infinite-horizon reach-avoid verification problem is formalized as follows.

\begin{problem}
\label{pro:ra}
Given a prescribed threshold $\epsilon\in(0,1)$, the infinite-horizon reach-avoid verification problem is to establish
\begin{equation}
\label{eq:uniform}
\mathbb{P}_{\mathrm{RA}}(\bm{x})\geq\epsilon,
\qquad
\forall\bm{x}\in\mathcal{X}_0.
\end{equation}
\end{problem}

We study a necessary and sufficient characterization of Problem \ref{pro:ra} in terms of continuous barrier functions. The following assumptions specify the system and set-theoretic conditions used throughout the paper.

\begin{assumption}
\label{ass}
\begin{enumerate}
\item For every $\bm{\theta}\in\Theta$, $\bm{f}(\cdot,\bm{\theta})$ is continuous on $\mathbb{R}^n$, and for every $\bm{x}\in\mathbb{R}^n$, $\bm{f}(\bm{x},\cdot)$ is measurable on $\Theta$.

\item The safe set $\mathcal{X}\subset\mathbb{R}^n$ and the target set $\mathcal{X}_r\subset\mathcal{X}$ are bounded and open.

\item The initial set $\mathcal{X}_0\subseteq\mathcal{X}\setminus\mathcal{X}_r$ is compact.

\item There exists a compact set $\widehat{\mathcal{X}}\subseteq\mathbb{R}^n$ such that
\begin{equation}
\label{xhat}
\overline{\mathcal{X}}
\cup
\bm{f}(\overline{\mathcal{X}},\Theta)
\subseteq
\widehat{\mathcal{X}},
\end{equation}
where $\bm{f}(\overline{\mathcal{X}},\Theta):=
\left\{
\bm{f}(\bm{x},\bm{\theta})\mid
\bm{x}\in\overline{\mathcal{X}},\
\bm{\theta}\in\Theta
\right\}$.
\end{enumerate}
\end{assumption}

A sufficient condition for the assumption 4) in Assumption \ref{ass} to hold is that $\bm{f}$ be jointly continuous in $(\bm{x},\bm{\theta})$ and $\Theta$ be compact: then $\overline{\mathcal{X}}\times\Theta$ is compact, and hence $\bm{f}(\overline{\mathcal{X}},\Theta)$ is compact. The compactness of $\widehat{\mathcal{X}}$ provides a convenient finite-domain setting for imposing barrier conditions and performing polynomial approximation. On the other hand, consider a trajectory starting from
$\mathcal{X}_0\subseteq\mathcal{X}$ and let $k$ be its first exit time from the safe set before reaching the target set $\mathcal{X}_r$. Then
$\bm{x}_{k-1}\in\overline{\mathcal{X}}$, and hence $\bm{x}_k=\bm{f}(\bm{x}_{k-1},\bm{\theta}_{k-1})\in
\bm{f}(\overline{\mathcal{X}},\Theta)
\subseteq \widehat{\mathcal{X}}$. Thus, any failure to satisfy the reach-avoid specification necessarily occurs through a one-step transition from $\mathcal{X}$ into $\widehat{\mathcal{X}}\setminus\mathcal{X}$ before reaching $\mathcal{X}_r$. Since such a first exit already constitutes a violation of the reach-avoid specification, the subsequent evolution of the trajectory is irrelevant. Consequently, it suffices to analyze the system within $\mathcal{X}$ and its one-step transitions into $\widehat{\mathcal{X}}\setminus\mathcal{X}$, rather than over the entire state space $\mathbb{R}^n$. For more details, please refer to \cite{xue2021reach,yu2023safe,xue2024finite,xue2026converse}.

Furthermore, we impose the following uniform absolute continuity condition on transition kernels.

For $\bm{x}\in\mathcal{X}$, define the one-step transition kernel
\[
\mathbb{P}(\bm{x},A)
:=
\mathbb{P}_{\bm{\theta}}
\bigl(
\bm{f}(\bm{x},\bm{\theta})\in A
\bigr),
\qquad
A\in\mathcal{B}(\widehat{\mathcal{X}}),
\]
where $\mathcal{B}(\widehat{\mathcal{X}})$ denotes the Borel $\sigma$-algebra on $\widehat{\mathcal{X}}$, i.e., $\mathcal{B}(\widehat{\mathcal{X}}):=\left\{
A\cap\widehat{\mathcal{X}} \mid 
A\in\mathcal{B}(\mathbb{R}^n)
\right\}$.

\begin{assumption}
\label{ass2}
The family of transition measures
\[
\left\{
\mathbb{P}(\bm{x},\cdot) \mid 
\bm{x}\in\mathcal{X}\setminus\mathcal{X}_r
\right\}
\]
is uniformly absolutely continuous with respect to the $n$-dimensional Lebesgue measure $\mu$ on $\widehat{\mathcal{X}}$. Specifically, for every $\rho>0$ there exists $\eta>0$ such that
\begin{equation}
\label{uac}
\mu(A)<\eta
\quad\Longrightarrow\quad
\sup_{\bm{x}\in\mathcal{X}\setminus\mathcal{X}_r}
\mathbb{P}(\bm{x},A)<\rho
\end{equation}
for every Borel set $A\subseteq\widehat{\mathcal{X}}$. Equivalently,
\[
\mu(A_j)\to 0
\quad\Longrightarrow\quad
\sup_{\bm{x}\in\mathcal{X}\setminus\mathcal{X}_r}
\mathbb{P}(\bm{x},A_j)\to0
\]
for every sequence $\{A_j\}_{j\ge1}$ of Borel subsets of $\widehat{\mathcal{X}}$.
\end{assumption}

Assumption \ref{ass2} requires that sets of sufficiently small Lebesgue measure have uniformly small transition probability. This is what allows us to control the effect of the small-measure exceptional sets arising when a bounded l.s.c. function is approximated by continuous ones.

A convenient sufficient condition for Assumption \ref{ass2} is that the transition kernel admit a uniformly bounded density with respect to $n$-dimensional Lebesgue measure. Suppose that, for every $\bm{x}\in\mathcal{X}\setminus\mathcal{X}_r$,
\[
\mathbb{P}(\bm{x},d\bm{y})
=
p(\bm{y}\mid\bm{x})\,d\bm{y},
\]
where
\[
\sup_{\bm{x}\in\mathcal{X}\setminus\mathcal{X}_r}
\|p(\cdot\mid\bm{x})\|_{L^\infty(\widehat{\mathcal{X}})}
\leq Q<\infty.
\]
Then, for every Borel set $A\subseteq\widehat{\mathcal{X}}$,
\[
\sup_{\bm{x}\in\mathcal{X}\setminus\mathcal{X}_r}
\mathbb{P}(\bm{x},A)
=
\sup_{\bm{x}\in\mathcal{X}\setminus\mathcal{X}_r}
\int_A p(\bm{y}\mid\bm{x})\,d\bm{y}
\leq Q\,\mu(A),
\]
so $\mu(A)\to0$ implies
$\sup_{\bm{x}\in\mathcal{X}\setminus\mathcal{X}_r}\mathbb{P}(\bm{x},A)\to0$, and Assumption \ref{ass2} holds. This condition is satisfied, for example, by additive-noise systems with uniformly bounded disturbance densities, and by state-dependent affine noise models
\[
\bm{f}(\bm{x},\bm{\theta})
=
\bm{g}(\bm{x})+B(\bm{x})\bm{\theta},
\]
provided the disturbance admits a bounded density and the noise matrices are uniformly nonsingular with uniformly bounded inverses over the relevant state region.

\subsection{Lower Semicontinuous Barrier Functions}
\label{sec:lsbc}

We now recall the necessary and sufficient barrier-like condition established in \cite{xue2026converse}. In particular, it guarantees the existence of a bounded l.s.c. barrier function whenever the reach-avoid specification holds for every state in the initial set with a strict margin.

\begin{proposition}[\cite{xue2026converse}]
\label{thm:1}
Under Assumption \ref{ass}, if there exist $\gamma\in(0,1)$ and a l.s.c. barrier function $v:\widehat{\mathcal{X}}\to\mathbb{R}$, bounded over $\mathcal{X}$, satisfying
\begin{equation}
\label{constraint_gamma0}
\begin{cases}
v(\bm{x})\geq\epsilon,
& \forall\bm{x}\in\mathcal{X}_0,\\[1mm]
v(\bm{x})
\leq
\gamma\,\mathbb{E}_{\bm{\theta}}
\left[
v(\bm{f}(\bm{x},\bm{\theta}))
\right],
& \forall\bm{x}\in\mathcal{X}\setminus\mathcal{X}_r,\\[1mm]
v(\bm{x})\leq1,
& \forall\bm{x}\in\mathcal{X}_r,\\[1mm]
v(\bm{x})\leq0,
& \forall\bm{x}\in\widehat{\mathcal{X}}\setminus\mathcal{X},
\end{cases}
\end{equation}
then $\mathbb{P}_{\mathrm{RA}}(\bm{x})\geq\epsilon$ for all $\bm{x}\in\mathcal{X}_0$. Furthermore, if the reach-avoid specification holds with a strict margin for every state in the initial set $\mathcal{X}_0$, i.e., $
\mathbb{P}_{\mathrm{RA}}(\bm{x})>\epsilon, \forall\bm{x}\in\mathcal{X}_0$, then there exist $\gamma\in(0,1)$ and a bounded l.s.c. barrier function $v:\mathbb{R}^n\to\mathbb{R}$ satisfying \eqref{constraint_gamma0}.
\end{proposition}

Proposition \ref{thm:1} establishes only the existence of an l.s.c. barrier function. As discussed in the introduction, such a function may exhibit irregularities that hinder numerical approximation and computational synthesis. This motivates the main question addressed in the next section: under Assumptions \ref{ass} and \ref{ass2}, can we establish a barrier-like condition involving continuous barrier functions that is both sufficient and necessary for infinite-horizon reach-avoid verification?

%% file: conditions.tex
\section{Continuous Barrier-like Conditions}
\label{sec:cbc}
In this section, we establish a necessary and sufficient barrier-like condition in terms of continuous barrier functions for infinite-horizon reach-avoid verification. Specifically, we show that the barrier-like condition in Theorem 4 of \cite{xue2024finite}, originally developed for finite-horizon reach-avoid verification, is sufficient for infinite-horizon reach-avoid verification. Under Assumptions \ref{ass}--\ref{ass2}, we then show that there also exists a continuous barrier function satisfying this condition with strict margins, thereby yielding necessity. Finally, these strict margins are used to establish the existence of polynomial barrier functions satisfying the condition with strict margins.

Before showing the necessary and sufficient barrier-like condition in terms of continuous barrier functions, we first strengthen this first non-strict condition in the barrier-like condition \eqref{constraint_gamma0} to a strict one.

\begin{lemma}
\label{strict}
Under Assumption \ref{ass}, if there exist a constant $\gamma\in(0,1)$ and a bounded l.s.c. barrier function
$v:\widehat{\mathcal{X}}\rightarrow\mathbb R$ satisfying
\begin{equation}
\label{constraint_gamma01}
\begin{cases}
v(\bm{x})>\epsilon,
& \forall\bm{x}\in\mathcal{X}_0,\\
v(\bm{x})
\leq
\gamma\mathbb{E}_{\bm{\theta}}
\left[
v(\bm f(\bm{x},\bm{\theta}))
\right],
& \forall\bm{x}\in\mathcal{X}\setminus\mathcal{X}_r,\\
v(\bm{x})\leq1,
& \forall\bm{x}\in\mathcal{X}_r,\\
v(\bm{x})\leq0,
& \forall\bm{x}\in\widehat{\mathcal{X}}\setminus\mathcal{X},
\end{cases}
\end{equation}
then $\mathbb P_{\mathrm{RA}}(\bm{x})\geq\epsilon, 
\forall\bm{x}\in\mathcal{X}_0$.
Furthermore, if the reach-avoid specification holds with a strict margin, i.e., $
\mathbb{P}_{\mathrm{RA}}(\bm{x})>\epsilon, \forall\bm{x}\in\mathcal{X}_0$,  there exist a constant
$\gamma\in(0,1)$ and a bounded l.s.c. barrier function
$v:\widehat{\mathcal{X}}\rightarrow\mathbb R$ satisfying \eqref{constraint_gamma01}.
\end{lemma}
\begin{pf}
Since $v(\bm{x})>\epsilon, \forall\bm{x}\in\mathcal{X}_0$ implies 
$v(\bm{x})\geq \epsilon, \forall\bm{x}\in\mathcal{X}_0$, the sufficiency conclusion can be obtained directly according to Theorem \ref{thm:1}. 

The necessity conclusion can be obtained by following the proof of Theorem 2 in \cite{xue2026converse}.
\end{pf}

Lemma \ref{strict} provides the strict margin needed for the subsequent continuity argument. Since $\mathcal{X}_0$ is compact and $v$ is l.s.c., the strict inequality implies $\inf_{\bm{x}\in\mathcal{X}_0}v(\bm{x})-\epsilon>0$. This positive margin allows us to perturb the l.s.c. barrier function from below without violating the initial-state condition. We next use this margin, together with the additional regularity of the transition kernel in Assumption \ref{ass2}, to construct a continuous barrier function.

The barrier-like condition for continuous barrier functions is given in Theorem 4 of \cite{xue2024finite}, where it was originally developed for finite-horizon reach-avoid verification. Compared with the barrier-like condition \eqref{constraint_gamma0} in Proposition \ref{thm:1}, this condition contains an additional additive constant $\beta$. When $\beta=0$, it reduces to the corresponding condition \eqref{constraint_gamma0} in Proposition \ref{thm:1}. We first show that, with the time horizon taken to infinity, this condition also provides a sufficient condition for infinite-horizon reach-avoid verification. Specifically, if a continuous barrier function satisfies the four inequalities in \eqref{caes:1}--\eqref{caes:4}, then the reach-avoid probability is at least $\epsilon$ for every initial state in $\mathcal{X}_0$. The main difficulty lies in establishing necessity. Under Assumption \ref{ass2} and the strict reach-avoid condition, Lemma \ref{strict} provides a bounded l.s.c. barrier function satisfying \eqref{constraint_gamma01}. We construct a continuous barrier function from this l.s.c. barrier by approximating it from below with an increasing sequence of continuous minorants. Egoroff's theorem then yields uniform convergence of this sequence outside a set of arbitrarily small Lebesgue measure. By the uniform absolute continuity condition in Assumption \ref{ass2}, the transition probability of this exceptional set is uniformly small. Consequently, for a sufficiently accurate continuous minorant, the error between the expected values of the l.s.c. barrier and its continuous minorant can be made uniformly small and absorbed into the strict margin provided by the additive constant $\beta$. A suitable constant shift of the continuous minorant then ensures the required strict inequalities on the initial, target, and exterior regions. Finally, continuity of the resulting one-step residual extends the inequality from $\mathcal{X}\setminus\mathcal{X}_r$ to its closure. This establishes the existence of a continuous barrier function satisfying all four inequalities in \eqref{caes:1}--\eqref{caes:4} strictly.

\begin{theorem}
\label{thm:con}
Under Assumption \ref{ass}, if there exist a continuous function
$\tilde v\in C(\widehat{\mathcal{X}})$,
$\gamma\in(0,1)$, and $\beta\in[0,\infty)$ satisfying
\begin{subnumcases}{}
\tilde{v}(\bm{x})
\geq
\displaystyle
\epsilon+\frac{\epsilon\beta-\beta}{1-\gamma},
& $\forall\bm{x}\in\mathcal{X}_0$,
\label{caes:1}\\
\beta+\tilde{v}(\bm{x})
\leq
\gamma\mathbb{E}_{\bm{\theta}}
\left[
\tilde{v}(\bm{f}(\bm{x},\bm{\theta}))
\right],
& $\forall\bm{x}\in
\overline{\mathcal{X}\setminus\mathcal{X}_r}$,
\label{caes:2}\\
\tilde{v}(\bm{x})\leq1,
& $\forall\bm{x}\in\overline{\mathcal{X}_r}$,
\label{caes:3}\\
(1-\gamma)\tilde{v}(\bm{x})\leq-\beta,
& $\forall\bm{x}\in
\overline{\widehat{\mathcal{X}}\setminus\mathcal{X}}$,
\label{caes:4}
\end{subnumcases}
then $\mathbb P_{\mathrm{RA}}(\bm{x})\geq\epsilon, \forall\bm{x}\in\mathcal{X}_0$. Moreover, under Assumption \ref{ass2}, if the reach-avoid specification holds with a strict margin, i.e., $
\mathbb{P}_{\mathrm{RA}}(\bm{x})>\epsilon, \forall\bm{x}\in\mathcal{X}_0$, there exist
a continuous function
$\tilde v\in C(\widehat{\mathcal{X}})$,
$\gamma\in(0,1)$, and $\beta\in[0,\infty)$ such that all four
inequalities in \eqref{caes:1}--\eqref{caes:4} hold strictly.
\end{theorem}

\begin{pf}
\noindent\textbf{1. Sufficiency Direction.}

The sufficiency follows from Theorem 4 in \cite{xue2024finite}.
Applying that result with
\[
\gamma_1=\frac{1}{\gamma},
\qquad
\beta_1=\frac{\beta}{\gamma},
\]
gives, for every $\bm{x}\in\mathcal{X}_0$,
\[
\begin{aligned}
\mathbb P_{\mathrm{RA}}(\bm{x})
&\geq
\lim_{N\rightarrow\infty}
\frac{
(\gamma_1^{N+1}\tilde v(\bm{x})-L)(\gamma_1-1)
+\beta_1(\gamma_1^{N+1}-1)
}{
(\gamma_1+\beta_1-1)(\gamma_1^{N+1}-1)
}\\
&=
\frac{(1-\gamma)\tilde v(\bm{x})+\beta}
{1-\gamma+\beta},
\end{aligned}
\]
where $\max_{\bm{x}\in \widehat{\mathcal{X}}}|\tilde{v}(\bm{x})| \leq L$. By \eqref{caes:1}, we have  $\tilde v(\bm{x})
\geq
\epsilon+\frac{\epsilon\beta-\beta}{1-\gamma}$,
and hence
\[
\begin{aligned}
(1-\gamma)\tilde v(\bm{x})+\beta
&\geq
\epsilon(1-\gamma)
+\epsilon\beta-\beta+\beta\\
&=
\epsilon(1-\gamma+\beta).
\end{aligned}
\]
Since $1-\gamma+\beta>0$,
\[
\mathbb P_{\mathrm{RA}}(\bm{x})
\geq
\frac{(1-\gamma)\tilde v(\bm{x})+\beta}
{1-\gamma+\beta}
\geq\epsilon.
\]
Therefore, $\mathbb P_{\mathrm{RA}}(\bm{x})\geq\epsilon, \forall\bm{x}\in\mathcal{X}_0$.

\medskip
\noindent\textbf{2. Necessity Direction.}

By Lemma \ref{strict}, there exist a constant
$\gamma\in(0,1)$ and a bounded l.s.c. barrier function
$v:\widehat{\mathcal{X}}\rightarrow\mathbb R$ satisfying
\eqref{constraint_gamma01}.

Define $m_0:=\inf_{\bm{x}\in\mathcal{X}_0}v(\bm{x})-\epsilon$. Since $v$ is l.s.c. and $\mathcal{X}_0$ is compact,
the infimum is attained. Moreover,
$v(\bm{x})>\epsilon$ on $\mathcal{X}_0$, and therefore $m_0>0$.

Choose $\beta\in(0,(1-\gamma)m_0)$,
and then choose $c$ such that
\[
\frac{\beta}{1-\gamma}<c<m_0.
\]
Define $r:=(1-\gamma)c-\beta>0$.

\medskip
\noindent\textbf{Step 1: Continuous Minorants.}

Since $v$ is bounded and l.s.c. on the compact set
$\widehat{\mathcal{X}}$, there exists a sequence
$\{w_j\}_{j\geq1}\subset C(\widehat{\mathcal{X}})$ such that
\[
w_j(\bm{x})\leq w_{j+1}(\bm{x})\leq v(\bm{x}),
\qquad
\forall\bm{x}\in\widehat{\mathcal{X}},
\]
and
\begin{equation}
\label{w_j}
w_j(\bm{x})\uparrow v(\bm{x}),
\qquad
\forall\bm{x}\in\widehat{\mathcal{X}}.
\end{equation}
For example, one may take
\[
w_j(\bm{x})
:=
\inf_{\bm{y}\in\widehat{\mathcal{X}}}
\left\{
v(\bm{y})+j\|\bm{x}-\bm{y}\|
\right\}.
\]

Let $M:=\sup_{\bm{y}\in\widehat{\mathcal{X}}}v(\bm{y}),
\quad
m:=\inf_{\bm{y}\in\widehat{\mathcal{X}}}v(\bm{y})$.
Then
\[
m\leq w_j(\bm{x})\leq v(\bm{x})\leq M,
\]
and hence
\[
0\leq v(\bm{x})-w_j(\bm{x})\leq M-m,
\qquad
\forall\bm{x}\in\widehat{\mathcal{X}}.
\]

\medskip
\noindent\textbf{Step 2: Choice of Approximation Parameters.}

Choose $\delta>0$ and $\rho>0$ sufficiently small such that
\[
\gamma\delta+\gamma(M-m)\rho<\frac{r}{2}.
\]
By Assumption \ref{ass2}, there exists $\eta>0$ such that $\mu(A)<\eta$ implies $\sup_{\bm{x}\in\mathcal{X}\setminus\mathcal{X}_r}
\mathbb P(\bm{x},A)<\rho$ for every Borel set
$A\subseteq\widehat{\mathcal{X}}$.

\medskip
\noindent\textbf{Step 3: Egoroff Approximation.}

Since $\widehat{\mathcal{X}}$ is compact, $\mu(\widehat{\mathcal{X}})<\infty$ holds.
Together with the pointwise convergence \eqref{w_j}, Egoroff's
theorem yields a measurable set
$E_0\subseteq\widehat{\mathcal{X}}$ such that $\mu(E_0)<\eta$
and $w_j\rightarrow v$ uniformly on
$\widehat{\mathcal{X}}\setminus E_0$.

By regularity of Lebesgue measure, choose a Borel set
$E\supseteq E_0$ satisfying
\[
\mu(E)<\eta.
\]
Then the convergence remains uniform on
$\widehat{\mathcal{X}}\setminus E$.

Hence, there exists $J_1$ such that, for all $j\geq J_1$,
\[
0\leq v(\bm{y})-w_j(\bm{y})\leq\delta,
\qquad
\forall\bm{y}\in\widehat{\mathcal{X}}\setminus E.
\]

\medskip
\noindent\textbf{Step 4: Uniform Initial-State Condition.}

Since $c<m_0$,
\[
v(\bm{x})
\geq
\epsilon+m_0
>
\epsilon+c,
\qquad
\forall\bm{x}\in\mathcal{X}_0.
\]
For each $j$, define
\[
U_j:=
\left\{
\bm{x}\in\mathcal{X}_0:
w_j(\bm{x})>\epsilon+c
\right\}.
\]
Each $U_j$ is relatively open in $\mathcal{X}_0$, and $U_j\subseteq U_{j+1}$. Moreover, by \eqref{w_j}, $\mathcal{X}_0=\bigcup_{j=1}^{\infty}U_j$.

Since $\mathcal{X}_0$ is compact, there exists $J_0$ such that $\mathcal{X}_0\subseteq U_{J_0}$. Thus, $w_{J_0}(\bm{x})>\epsilon+c, \forall\bm{x}\in\mathcal{X}_0$.

Choose $j\geq\max\{J_0,J_1\}$,
and define
\[
\tilde v:=w_j-c.
\]
Then $\tilde v\in C(\widehat{\mathcal{X}})$
and $\tilde v(\bm{x})>\epsilon, \forall\bm{x}\in\mathcal{X}_0$.

Moreover,
\begin{equation}
\label{uni:delta}
0\leq v(\bm{y})-w_j(\bm{y})\leq\delta,
\qquad
\forall\bm{y}\in\widehat{\mathcal{X}}\setminus E.
\end{equation}

\medskip
\noindent\textbf{Step 5: Uniform Transition Error.}

For any
$\bm{x}\in\mathcal{X}\setminus\mathcal{X}_r$,
\[
\begin{aligned}
0
&\leq
\mathbb{E}_{\bm{\theta}}
\left[
v(\bm f(\bm{x},\bm{\theta}))
-w_j(\bm f(\bm{x},\bm{\theta}))
\right]\\
&=
\int_{\widehat{\mathcal{X}}}
(v(\bm{y})-w_j(\bm{y}))
\mathbb P(\bm{x},d\bm{y})\\
&\leq
\delta\mathbb P(\bm{x},\widehat{\mathcal{X}}\setminus E)
+
(M-m)\mathbb P(\bm{x},E)\\
&\leq
\delta+(M-m)\mathbb P(\bm{x},E)\\
&<
\delta+(M-m)\rho.
\end{aligned}
\]
Therefore,
\begin{equation}
\label{cc}
\gamma
\sup_{\bm{x}\in\mathcal{X}\setminus\mathcal{X}_r}
\left|
\mathbb{E}_{\bm{\theta}}
[v(\bm f(\bm{x},\bm{\theta}))]
-\mathbb{E}_{\bm{\theta}}
[w_j(\bm f(\bm{x},\bm{\theta}))]
\right|
<\frac r2.
\end{equation}

\medskip
\noindent\textbf{Step 6: Bellman Inequality.}

For
$\bm{x}\in\mathcal{X}\setminus\mathcal{X}_r$,
\[
w_j(\bm{x})
\leq
v(\bm{x})
\leq
\gamma
\mathbb{E}_{\bm{\theta}}
[v(\bm f(\bm{x},\bm{\theta}))].
\]
Consequently,
\[
\begin{aligned}
&\gamma
\mathbb{E}_{\bm{\theta}}
[\tilde v(\bm f(\bm{x},\bm{\theta}))]
-\tilde v(\bm{x})\\
=&
\gamma
\mathbb{E}_{\bm{\theta}}
[w_j(\bm f(\bm{x},\bm{\theta}))]
-w_j(\bm{x})
+(1-\gamma)c\\
\geq&
-\gamma
\left|
\mathbb{E}_{\bm{\theta}}
[v(\bm f(\bm{x},\bm{\theta}))]
-
\mathbb{E}_{\bm{\theta}}
[w_j(\bm f(\bm{x},\bm{\theta}))]
\right|
+(1-\gamma)c\\
> &
-\frac r2+(1-\gamma)c=\beta+\frac r2.
\end{aligned}
\]
Hence
\[
\gamma
\mathbb{E}_{\bm{\theta}}
[\tilde v(\bm f(\bm{x},\bm{\theta}))]
-\tilde v(\bm{x})
>
\beta+\frac r2
\]
for all
$\bm{x}\in\mathcal{X}\setminus\mathcal{X}_r$.

Since $\tilde v$ is continuous on the compact set
$\widehat{\mathcal{X}}$, it is bounded. Moreover,
$\bm f(\cdot,\bm{\theta})$ is continuous for every $\bm{\theta}$.
Therefore, by dominated convergence,
\[
\bm{x}\mapsto
\mathbb{E}_{\bm{\theta}}
[\tilde v(\bm f(\bm{x},\bm{\theta}))]
\]
is continuous over $\overline{\mathcal{X}}$. Hence $g(\bm{x}):= \gamma
\mathbb{E}_{\bm{\theta}}
[\tilde v(\bm f(\bm{x},\bm{\theta}))]
-\tilde v(\bm{x})$
is continuous over $\overline{\mathcal{X}}$.

Since $g(\bm{x})>\beta+\frac r2$ on $\mathcal{X}\setminus\mathcal{X}_r$, continuity implies
\[
g(\bm{x})\geq\beta+\frac r2>\beta,
\qquad
\forall\bm{x}\in
\overline{\mathcal{X}\setminus\mathcal{X}_r}.
\]
Thus, $\beta+\tilde v(\bm{x})< \gamma \mathbb{E}_{\bm{\theta}}
[\tilde v(\bm f(\bm{x},\bm{\theta}))], \forall\bm{x}\in
\overline{\mathcal{X}\setminus\mathcal{X}_r}$.

\medskip
\noindent\textbf{Step 7: Target-Set Condition.}

For $\bm{x}\in\mathcal{X}_r$,
\[
\tilde v(\bm{x})
=
w_j(\bm{x})-c
\leq
v(\bm{x})-c
\leq
1-c<1.
\]
By continuity of $\tilde v$, $\tilde v(\bm{x})\leq1-c<1, \forall\bm{x}\in\overline{\mathcal{X}_r}$.

\medskip
\noindent\textbf{Step 8: Exterior Condition.}

For
$\bm{x}\in\widehat{\mathcal{X}}\setminus\mathcal{X}$,
\[
\tilde v(\bm{x})
=
w_j(\bm{x})-c
\leq
-c.
\]
Thus, $(1-\gamma)\tilde v(\bm{x}) \leq -(1-\gamma)c <-\beta$.
By continuity,
\[
(1-\gamma)\tilde v(\bm{x})
\leq
-(1-\gamma)c
<-\beta,
\qquad
\forall\bm{x}\in
\overline{\widehat{\mathcal{X}}\setminus\mathcal{X}}.
\]

Finally, since $\epsilon \in (0,1)$ and $\beta\geq0$,
\[
\frac{\epsilon\beta-\beta}{1-\gamma}
=
-\frac{(1-\epsilon)\beta}{1-\gamma}
\leq0.
\]
Therefore, $\tilde v(\bm{x})>\epsilon \geq \epsilon+\frac{\epsilon\beta-\beta}{1-\gamma},
\qquad
\forall\bm{x}\in\mathcal{X}_0$.

Hence all four inequalities in
\eqref{caes:1}--\eqref{caes:4} hold strictly.
\end{pf}

Lemma \ref{thm:con} establishes a continuous counterpart of the l.s.c. converse characterization. Under Assumptions \ref{ass} and \ref{ass2}, whenever the reach-avoid specification holds with a strict margin for every state in the initial set, there exists a continuous barrier function satisfying the barrier-like condition \eqref{caes:1}--\eqref{caes:4}, with strict inequalities on all relevant regions. These strict margins provide the robustness needed for a further regularity step. Since polynomials are dense in the space of continuous functions on compact sets, the continuous barrier function can be uniformly approximated by a polynomial, leading to the existence of polynomial barrier functions.

\begin{theorem}[Polynomial Barrier Functions]
\label{thm:polynomial}
Under Assumption \ref{ass}, if there exist a polynomial
$p\in\mathbb R[\bm{x}]$,
$\gamma\in(0,1)$, and
$\beta\in[0,\infty)$ satisfying the barrier-like condition \eqref{caes:1}--\eqref{caes:4},
then $\mathbb P_{\mathrm{RA}}(\bm{x})\geq \epsilon, \forall\bm{x}\in\mathcal{X}_0$.
Moreover, under Assumption \ref{ass2}, if the reach-avoid specification holds with a strict margin for every state in the initial set $\mathcal{X}_0$, i.e., $
\mathbb{P}_{\mathrm{RA}}(\bm{x})>\epsilon, \forall\bm{x}\in\mathcal{X}_0$, there exist a
polynomial $p\in\mathbb R[\bm{x}]$,
$\gamma\in(0,1)$, and $\beta\in[0,\infty)$ such that all four
inequalities in \eqref{caes:1}--\eqref{caes:4} hold strictly.
\end{theorem}

\begin{pf}
\noindent\textbf{1. Sufficiency Direction.}

Since every polynomial is continuous on $\widehat{\mathcal{X}}$, by Theorem \ref{thm:con}, we have $\mathbb P_{\mathrm{RA}}(\bm{x})\geq\epsilon, \forall\bm{x}\in\mathcal{X}_0$.

\medskip
\noindent\textbf{2. Necessity Direction.}

The necessity proof adopts the margin-based approximation idea of Corollary 3 in \cite{xue2019inner} and Lemma 5 in \cite{li2026converse}: a sufficiently accurate approximation is combined with strict margins to absorb the approximation error and preserve the barrier-like inequalities.

By Theorem \ref{thm:con}, under Assumptions \ref{ass} and
\ref{ass2}, there exist
$\gamma\in(0,1)$, $\beta>0$, and
$\tilde v\in C(\widehat{\mathcal{X}})$ satisfying all four 
barrier-like inequalities \eqref{caes:1}--\eqref{caes:4} strictly.

In particular, using the construction in the proof of
Theorem \ref{thm:con}, choose $\beta\in(0,(1-\gamma)m_0)$,
then choose $\frac{\beta}{1-\gamma}<c<m_0$,
and define
\[
r:=(1-\gamma)c-\beta>0.
\]
The approximation parameters can be chosen such that $\gamma\delta_0+\gamma(M-m)\rho<\frac r2$.

The resulting continuous function $\tilde v$ satisfies
\[
\begin{cases}
\tilde v(\bm{x})>\epsilon,& \forall\bm{x}\in\mathcal{X}_0,\\
\gamma
\mathbb{E}_{\bm{\theta}}
[\tilde v(\bm f(\bm{x},\bm{\theta}))]
-\tilde v(\bm{x})
>
\beta+\frac r2, &
\forall\bm{x}\in
\overline{\mathcal{X}\setminus\mathcal{X}_r},\\
\tilde v(\bm{x})\leq1-c<1, &\forall\bm{x}\in\overline{\mathcal{X}_r},\\
(1-\gamma)\tilde v(\bm{x})
\leq-\beta-r<-\beta,&
\forall\bm{x}\in
\overline{\widehat{\mathcal{X}}\setminus\mathcal{X}}.
\end{cases}
\]

Since $\mathcal{X}_0$ is compact and $\tilde v>\epsilon$ on
$\mathcal{X}_0$, there exists $s_0>0$ such that $\tilde v(\bm{x})\geq\epsilon+s_0, \forall\bm{x}\in\mathcal{X}_0$.

By the Stone--Weierstrass theorem, for every $\delta>0$ there exists
a polynomial $p\in\mathbb R[\bm{x}]$ satisfying $\|p-\tilde v\|_{\infty,\widehat{\mathcal{X}}}<\delta$.  Choose $0<\delta<\min\left\{\frac{s_0}{2},\frac{r}{2(1+\gamma)},c, \frac{r}{1-\gamma}\right\}$.

Since $\mathbb P(\bm{x},\cdot)$ is a probability measure,
\[
\begin{aligned}
&\left|
\mathbb{E}_{\bm{\theta}}
[p(\bm f(\bm{x},\bm{\theta}))]
-
\mathbb{E}_{\bm{\theta}}
[\tilde v(\bm f(\bm{x},\bm{\theta}))]
\right|\\
\leq &
\int_{\widehat{\mathcal{X}}}
|p(\bm{y})-\tilde v(\bm{y})|
\mathbb P(\bm{x},d\bm{y})\\
\leq &
\|p-\tilde v\|_{\infty,\widehat{\mathcal{X}}}
<\delta, \qquad \forall \bm{x}\in \overline{\mathcal{X}\setminus \mathcal{X}_r}.
\end{aligned}
\]

For $\bm{x}\in\mathcal{X}_0$, $p(\bm{x})>\epsilon+\frac{s_0}{2}
>\epsilon$. Since $\frac{\epsilon\beta-\beta}{1-\gamma}=
-\frac{(1-\epsilon)\beta}{1-\gamma}\leq 0$,
we obtain
\[
p(\bm{x})
>
\epsilon
>
\epsilon+
\frac{\epsilon\beta-\beta}{1-\gamma},
\qquad
\forall\bm{x}\in\mathcal{X}_0.
\]
Thus \eqref{caes:1} holds strictly.

For
$\bm{x}\in\overline{\mathcal{X}\setminus\mathcal{X}_r}$,
\[
\begin{aligned}
&\gamma
\mathbb{E}_{\bm{\theta}}
[p(\bm f(\bm{x},\bm{\theta}))]
-p(\bm{x})\\
\geq &
\gamma
\mathbb{E}_{\bm{\theta}}
[\tilde v(\bm f(\bm{x},\bm{\theta}))]
-\tilde v(\bm{x})
-(1+\gamma)\delta\\
> &
\beta+\frac r2-(1+\gamma)\delta >\beta,
\end{aligned}
\]
implying $\beta+p(\bm{x})<\gamma\mathbb{E}_{\bm{\theta}}
[p(\bm f(\bm{x},\bm{\theta}))]$ for all
$\bm{x}\in\overline{\mathcal{X}\setminus\mathcal{X}_r}$.
Hence \eqref{caes:2} holds strictly.

For $\bm{x}\in\overline{\mathcal{X}_r}$,
\[
p(\bm{x})
\leq
\tilde v(\bm{x})+\delta
\leq
1-c+\delta
<1,
\]
because $\delta<c$. Hence \eqref{caes:3} holds strictly.

Finally, for
$\bm{x}\in
\overline{\widehat{\mathcal{X}}\setminus\mathcal{X}}$,
\[
\begin{aligned}
(1-\gamma)p(\bm{x})
&\leq
(1-\gamma)\tilde v(\bm{x})
+(1-\gamma)\delta\\
&\leq
-\beta-r+(1-\gamma)\delta\\
&<-\beta,
\end{aligned}
\]
implying $(1-\gamma)p(\bm{x})<-\beta, \forall\bm{x}\in
\overline{\widehat{\mathcal{X}}\setminus\mathcal{X}}$. Thus \eqref{caes:4} holds strictly.

Therefore, there exist a polynomial
$p\in\mathbb R[\bm{x}]$, $\gamma\in(0,1)$, and
$\beta\in [0,\infty)$ satisfying
\eqref{caes:1}--\eqref{caes:4} strictly.
\end{pf}

%% file: sos.tex
\section{Computation via Sum-of-Squares Programming}
\label{sec:sos}

The polynomial barrier functions established in Theorem \ref{thm:polynomial} provide a finite-dimensional representation of the continuous barrier functions satisfying \eqref{caes:1}--\eqref{caes:4} strictly. In this section, we develop a computational procedure for synthesizing such functions for polynomial systems using SOS programming. The key idea is to express the relevant state-space regions as compact basic semialgebraic sets and to replace polynomial positivity conditions with SOS certificates, i.e., representations of positive polynomials in terms of sums of squares and the defining polynomials of the corresponding regions. Under the Archimedean condition, Putinar's Positivstellensatz guarantees the completeness of this positivity certification as the degrees of the SOS multipliers increase.

\subsection{Algebraic Representations and Computational Assumptions}

We assume that the four state-space regions appearing in the barrier-like condition \eqref{caes:1}--\eqref{caes:4} admit compact basic semialgebraic representations:
\[
\begin{cases}
\mathcal{X}_0=
\left\{
\bm{x}\in\mathbb{R}^n
\mid
g_{0,i}(\bm{x})\geq0,\;
i=1,\ldots,k_0
\right\},\\
\overline{\mathcal{X}\setminus\mathcal{X}_r}=
\left\{
\bm{x}\in\mathbb{R}^n
\mid
g_{u,i}(\bm{x})\geq0,\;
i=1,\ldots,k_u
\right\},\\
\overline{\mathcal{X}_r}=
\left\{
\bm{x}\in\mathbb{R}^n
\mid
g_{r,i}(\bm{x})\geq0,\;
i=1,\ldots,k_r
\right\},\\
\overline{\widehat{\mathcal{X}}\setminus\mathcal{X}}=
\left\{
\bm{x}\in\mathbb{R}^n
\mid
g_{e,i}(\bm{x})\geq0,\;
i=1,\ldots,k_e
\right\},
\end{cases}
\]
where $g_{0,i}, g_{u,i}, g_{r,i}, g_{e,i} \in\mathbb{R}[\bm{x}]$
are known polynomials.

For a prescribed degree $d\in\mathbb{N}$, we restrict the search for barrier functions to the finite-dimensional polynomial space
\[
\mathcal{P}_d:=
\left\{
p\in\mathbb{R}[\bm{x}]
\mid
\deg(p)\leq d
\right\}.
\]
The existence result in Theorem \ref{thm:polynomial} guarantees the existence of a polynomial barrier function but does not, by itself, ensure that its one-step expectation can be represented as a polynomial. For computational synthesis, we thus impose the following assumption on the system dynamics and transition operator.

\begin{assumption}
\label{ass:poly}
For every $\bm{\theta}\in\Theta$, the mapping $\bm{x}\mapsto\bm{f}(\bm{x},\bm{\theta})$ is polynomial. Moreover, for the prescribed degree $d$, the
transition operator
\[
(\mathbb{P}p)(\bm{x})
:=
\mathbb{E}_{\bm{\theta}}
\left[
p\bigl(\bm{f}(\bm{x},\bm{\theta})\bigr)
\right]
\]
maps $\mathcal{P}_d$ into $\mathbb{R}[\bm{x}]$; that is, $p\in\mathcal{P}_d \quad\Longrightarrow\quad
\mathbb{P}p\in\mathbb{R}[\bm{x}]$.
\end{assumption}

Under Assumption \ref{ass}, for any $p\in\mathcal{P}_d$, the transition term $\mathbb{P}p$ is a polynomial. Consequently, the one-step residual $\gamma(\mathbb{P}p)(\bm{x})-p(\bm{x})-\beta$ is a polynomial. Hence, each polynomial barrier-like inequality can
be formulated as a polynomial positivity condition on the
corresponding semialgebraic region.

We next introduce the SOS framework used to certify these positivity
conditions. Let
\[
\Sigma[\bm{x}]:=
\left\{
q\in\mathbb{R}[\bm{x}]
\mid
q(\bm{x})=\sum_j h_j^2(\bm{x}),
\quad
h_j\in\mathbb{R}[\bm{x}]
\right\}
\]
denote the cone of SOS polynomials. For a collection of defining
polynomials $G=\{g_1,\ldots,g_k\}$, the associated quadratic module is
\[
\mathcal{M}(G):=
\left\{
\sigma_0(\bm{x})+
\sum_{i=1}^{k}\sigma_i(\bm{x})g_i(\bm{x})
\;\middle|\;
\sigma_0,\ldots,\sigma_k\in\Sigma[\bm{x}]
\right\}.
\]
A quadratic module $\mathcal{M}(G)$ is called \emph{Archimedean} if
there exists a constant $R>0$ such that
\[
R-\|\bm{x}\|_2^2\in\mathcal{M}(G).
\]
The Archimedean property enables the application of Putinar's
Positivstellensatz, which guarantees SOS certificates for polynomials
that are strictly positive on the corresponding semialgebraic set.

We next show how this property can be ensured for the four regions
considered in this paper. Since $\widehat{\mathcal{X}}$ is
compact, there exists $R>0$ such that
\[
\widehat{\mathcal{X}}
\subseteq
\left\{
\bm{x}\in\mathbb{R}^n
\mid
R-\|\bm{x}\|_2^2\geq0
\right\}.
\]
Define the redundant ball constraint
\[
g_{\mathrm b}(\bm{x}):=R-\|\bm{x}\|_2^2.
\]
Because each of the four regions is contained in
$\widehat{\mathcal{X}}$, the constraint $g_{\mathrm b}(\bm{x})\geq0$
holds throughout each region and therefore does not change any of the
four sets. We can thus augment the defining polynomials of each region
with $g_{\mathrm b}$. Specifically, let
\[
G_0=\{g_{0,1},\ldots,g_{0,k_0},g_{\mathrm b}\},
\qquad
G_u=\{g_{u,1},\ldots,g_{u,k_u},g_{\mathrm b}\},
\]
and similarly define
\[
G_r=\{g_{r,1},\ldots,g_{r,k_r},g_{\mathrm b}\},
\qquad
G_e=\{g_{e,1},\ldots,g_{e,k_e},g_{\mathrm b}\}.
\]
By construction, $R-\|\bm{x}\|_2^2= g_{\mathrm b}(\bm{x})
\in\mathcal{M}(G_\ell),
\qquad
\ell\in\{0,u,r,e\}$.
Hence, all four quadratic modules are Archimedean.

The following version of Putinar's Positivstellensatz
\cite{putinar1993positive} now provides the positivity certificates
used in our SOS formulation.

\begin{theorem}[Putinar's Positivstellensatz \cite{putinar1993positive}]
\label{thm:putinar}
Let
\[
G=
\left\{
\bm{x}\in\mathbb{R}^n
\;\middle|\;
g_i(\bm{x})\geq0,\;
i=1,\ldots,k
\right\}
\]
be a basic closed semialgebraic set. Suppose that $\mathcal{M}(G)$ is Archimedean. Then, for every polynomial
$p\in\mathbb{R}[\bm{x}]$ satisfying $p(\bm{x})>0, \forall\,\bm{x}\in G$, there exist SOS polynomials
$\sigma_0,\sigma_1,\ldots,\sigma_k\in\Sigma[\bm{x}]$ such that
\[
p(\bm{x})=\sigma_0(\bm{x})+\sum_{i=1}^{k} \sigma_i(\bm{x})g_i(\bm{x}).
\]
Equivalently, $p\in\mathcal{M}(G)$.
\end{theorem}

Thus, for a fixed polynomial degree and fixed multiplier degree, the resulting SOS conditions lead to a finite-dimensional SDP. Moreover, since the corresponding quadratic modules are Archimedean, Putinar's Positivstellensatz guarantees that every polynomial that is strictly positive on one of these regions admits an SOS representation at some finite degree. Consequently, the SOS hierarchy is complete in the sense that every strictly positive polynomial constraint is certified at a finite hierarchy level. They are formally shown in the sequel. The above results have also been widely applied in the existing reachability analysis literature, e.g., \cite{henrion2013convex,xue2019inner}.

\subsection{SOS Reformulation and Soundness}

We now translate the four polynomial barrier-like conditions into SOS
constraints. Fix
\[
\gamma\in(0,1),
\qquad
\epsilon\in[0,1).
\]
Define the polynomial residuals
\begin{align}
q_0(\bm{x})
&:=
p(\bm{x})
-
\left(
\epsilon+
\frac{\epsilon\beta-\beta}{1-\gamma}
\right),
\label{eq:q0-sos}\\
q_u(\bm{x})
&:=
\gamma\mathbb{E}_{\bm{\theta}}
\left[
p(\bm{f}(\bm{x},\bm{\theta}))
\right]
-p(\bm{x})-\beta,
\label{eq:qu-sos}\\
q_r(\bm{x})
&:=1-p(\bm{x}),
\label{eq:qr-sos}\\
q_e(\bm{x})
&:=
-\beta-(1-\gamma)p(\bm{x}).
\label{eq:qe-sos}
\end{align}

The conditions in Theorem \ref{thm:polynomial} are precisely
positivity conditions on these four residual polynomials. We therefore seek $\beta\geq 0$ and SOS multiplier polynomials
\[
\sigma_{0,i},\quad
\sigma_{u,i},\quad
\sigma_{r,i},\quad
\sigma_{e,i}
\in\Sigma[\bm{x}]
\]
such that
\begin{subnumcases}{}
q_0(\bm{x})-
\displaystyle\sum_{i=1}^{k_0+1}
\sigma_{0,i}(\bm{x})g_{0,i}(\bm{x})
\in\Sigma[\bm{x}],
\label{sos:1}\\[2mm]
q_u(\bm{x})
-
\displaystyle\sum_{i=1}^{k_u+1}
\sigma_{u,i}(\bm{x})g_{u,i}(\bm{x})
\in\Sigma[\bm{x}],
\label{sos:2}\\[2mm]
q_r(\bm{x})
-
\displaystyle\sum_{i=1}^{k_r+1}
\sigma_{r,i}(\bm{x})g_{r,i}(\bm{x})
\in\Sigma[\bm{x}],
\label{sos:3}\\[2mm]
q_e(\bm{x})
-
\displaystyle\sum_{i=1}^{k_e+1}
\sigma_{e,i}(\bm{x})g_{e,i}(\bm{x})
\in\Sigma[\bm{x}],
\label{sos:4}
\end{subnumcases}
where $g_{\ell,k_\ell+1}=g_{\mathrm{b}}, \ell\in\{0,u,r,e\}.$

Equivalently, one may explicitly introduce the zeroth-order SOS
multipliers and write
\[
q_\ell(\bm{x})
=
\sigma_{\ell,0}(\bm{x})
+
\sum_{i=1}^{k_\ell+1}
\sigma_{\ell,i}(\bm{x})g_{\ell,i}(\bm{x}),
\qquad
\ell\in\{0,u,r,e\}.
\]

The following theorem shows that feasibility of these SOS constraints
is sufficient for infinite-horizon reach-avoid verification.

\begin{theorem}[Soundness of the SOS Certificate]
\label{thm:sos_soundness}
Suppose Assumption \ref{ass:poly} holds. Fix
$\gamma\in(0,1)$ and $\beta\geq 0$. If there exist
$p\in\mathcal{P}_d$ and SOS multipliers satisfying
\eqref{sos:1}--\eqref{sos:4}, then $p$ satisfies the 
barrier-like conditions \eqref{caes:1}--\eqref{caes:4}. Consequently,
$\mathbb{P}_{\mathrm{RA}}(\bm{x}) \geq \epsilon, \forall\bm{x}\in\mathcal{X}_0$.
\end{theorem}

\begin{pf}
For any $\bm{x}\in\mathcal{X}_0$,
\[
g_{0,i}(\bm{x})\geq0,
\qquad
\sigma_{0,i}(\bm{x})\geq 0.
\]
Hence $\sum_{i=1}^{k_0+1}\sigma_{0,i}(\bm{x})g_{0,i}(\bm{x}) \geq 0$.

Since the left-hand side of \eqref{sos:1} is an SOS polynomial,
it is nonnegative. Therefore, $q_0(\bm{x})\geq0, \forall\bm{x}\in\mathcal{X}_0$, which is precisely \eqref{caes:1}. The same argument applied to
\eqref{sos:2}--\eqref{sos:4} establishes
\eqref{caes:2}--\eqref{caes:4}. The reach-avoid guarantee then
follows from Theorem \ref{thm:con}.
\end{pf}

\subsection{Completeness of the SOS Formulation}
In this section, we next show that the SOS formulation does not introduce additional conservatism at the level of polynomial barrier functions when sufficiently high degrees are allowed. The key ingredient is Putinar's Positivstellensatz together with the strict polynomial positivity guaranteed by Theorem~\ref{thm:polynomial}.

\begin{theorem}[Completeness of SOS Optimization]
\label{thm:sos_completeness}
Suppose Assumptions \ref{ass} -- \ref{ass:poly} hold, and . If the reach-avoid specification holds with a strict margin, i.e., $
\mathbb{P}_{\mathrm{RA}}(\bm{x})>\epsilon, \forall\bm{x}\in\mathcal{X}_0$, then there exist a
polynomial $p\in\mathbb{R}[\bm{x}]$, a constant
$\gamma\in(0,1)$, a constant $\beta>0$, and finite-degree SOS
multipliers satisfying \eqref{sos:1}--\eqref{sos:4}, provided that the polynomial and multiplier degrees are chosen sufficiently large.
\end{theorem}
\begin{pf}
By Theorem \ref{thm:polynomial}, under Assumptions \ref{ass} and
\ref{ass2}, there exist a polynomial
$p\in\mathbb{R}[\bm{x}]$, a constant
$\gamma\in(0,1)$, and a constant $\beta>0$ such that all four
polynomial barrier-like inequalities hold strictly on their
respective sets.

Let  $p\in\mathcal{P}_d$. By Assumption \ref{ass:poly}, $(\mathbb{P}p)(\bm{x}) = \mathbb{E}_{\bm{\theta}}
\left[
p(\bm f(\bm{x},\bm{\theta}))
\right]
$
is a polynomial, and hence the residuals
$q_0,q_u,q_r,q_e$ defined in
\eqref{eq:q0-sos}--\eqref{eq:qe-sos} are all polynomials.

Because each corresponding set is compact and the associated residual is strictly positive, by Putinar's Positivstellensatz, $q_\ell
\in\mathcal{M}(G_\ell)$ for $\ell\in\{0,u,r,e\}$.

Since the polynomial $p$ has finite degree and each Putinar
representation uses finitely many polynomial multipliers, sufficiently large polynomial and multiplier degree bounds make the SOS problem feasible.
\end{pf}

%% file: examples.tex
\section{Illustrative Examples}
\label{sec:ex}
In this section, we demonstrate the application of our theoretical developments through two examples. In these two cases, $\bm{f}(\bm{x},\bm{\theta})$ is a polynomial, and the initial set $\mathcal{X}_0$, the safe set $\mathcal{X}$, and the target set $\mathcal{X}_r$ are semialgebraic sets. We aim to search for polynomial barrier functions that satisfy the SOS constraints \eqref{sos:1} -- \eqref{sos:4}. For a fixed $\gamma\in(0,1)$ and $\epsilon \in (0,1)$, the resulting SOS program is a convex SDP in the coefficients of $p$, the scalar $\beta$, and the Gram matrices of the SOS multipliers. As established in \cite{xue2026sufficient,xue2026converse}, the barrier-like condition becomes less conservative as $\gamma$ approaches $1$. This suggests a practical computational strategy of fixing $\gamma$ and solving a sequence of convex SDPs for values of $\gamma$ increasingly close to $1$. In the experiments, the resulting SDPs are  solved  using the tool Mosek 10.1.21 \cite{aps2019mosek}. To ensure numerical stability during the solution of these SDPs, we additionally impose a constraint on the coefficients of the unknown polynomials, specifically restricting them to the interval $[-10^2, 10^2]$.

In addition, to empirically validate the reach-avoid guarantees, we employ a parallelized Monte Carlo approach. We sample $10^4$ initial states from $\mathcal{X}_0$ using a polar grid and simulate $10^4$ independent trajectories from each initial state for $10^3$ steps under uniformly distributed disturbances. Each trajectory is monitored to determine whether it reaches the target set $\mathcal{X}_r$ before leaving the safe set $\mathcal{X}$. For each sampled initial state, we compute the empirical reach-avoid probability and estimate the worst-case probability by taking the minimum over the sampled initial states.

\begin{example}
\label{ex1}
We consider the following system, adapted from Example 1 in \cite{xue2021reach},
\begin{equation*}
\begin{cases}
x(i+1)=x(i)+0.01\left(\begin{split}&-0.5x(i)-0.5y(i)\\
&+0.5x(i)y(i)+\theta_1(i)\end{split}\right),\\
y(i+1)=y(i)+0.01(-0.5y(i)+1+\theta_2(i)),
\end{cases}
\end{equation*}
where $\theta_1(i)\in\Theta_1=[-0.01,0.01]$ and $\theta_2(i)\in\Theta_2=[-10,10]$ are uniform, $\mathcal{X}_0=\{\bm{x}\mid g_{0,1}(x,y) \geq 0\}$ with $g_{0,1}(\bm{x})=-x^2-(y+0.7)^2+0.01$, $\mathcal{X}=\{\bm{x}\mid h(\bm{x})> 0\}$ with $h(\bm{x})=-x^2-y^2+1$,  $\mathcal{X}_r=\{\bm{x}\mid g_{r,1}(\bm{x})>0\}$ with $g_{r,1}(\bm{x})=-10x^2-10(y-0.5)^2+1$, and $\widehat{\mathcal{X}}=\{\bm{x}\mid \widehat{h}(\bm{x})\geq 0\}$ with $\widehat{h}(\bm{x})=-x^2-y^2+1.1$.  The Monte Carlo (\textbf{MC}) simulations yield an empirical worst-case reach-avoid probability of $0.8445$ over the sampled initial states. Three system trajectories, originating in the initial set $\mathcal{X}_0$, are visualized on Fig. \ref{fig:ex1_2}.

It is evident that this example satisfies Assumptions \ref{ass} -- \ref{ass:poly}. Also, this example satisfies Assumption \ref{ass2}. Indeed, the disturbance
$\bm{\theta}=(\theta_1,\theta_2)$ is uniformly distributed on $\Theta=[-0.01,0.01]\times[-10,10]$, with density $q(\bm{\theta})=2.5$. Moreover,
\[
\nabla_{\bm{\theta}}\bm{f}(\bm{x},\bm{\theta})
=
\begin{bmatrix}
0.01 & 0\\
0 & 0.01
\end{bmatrix},
\qquad
\left|\det \nabla_{\bm{\theta}}\bm{f}(\bm{x},\bm{\theta})\right|
=
10^{-4}.
\]
Hence, the sufficient condition for Assumption \ref{ass2} holds with
$Q=2.5$ and $c_J=10^{-4}$, and the transition kernel admits a uniformly
bounded density satisfying
\[
\sup_{\bm{x}\in\mathcal{X}\setminus\mathcal{X}_r}
\|p(\cdot\mid\bm{x})\|_{L^\infty(\widehat{\mathcal{X}})}
\leq
\frac{Q}{c_J}
=
25000.
\]
Therefore, Assumption \ref{ass2} is satisfied.

The verification results via solving SOS program \eqref{sos:1}--\eqref{sos:4} with $g_{\mathrm{b}}=-x^2-y^2+1.1$ are shown in Table \ref{tab:sdp_ex1_set} 

\begin{figure}[h]
\centering
\includegraphics[width=0.4\textwidth, height=5cm]{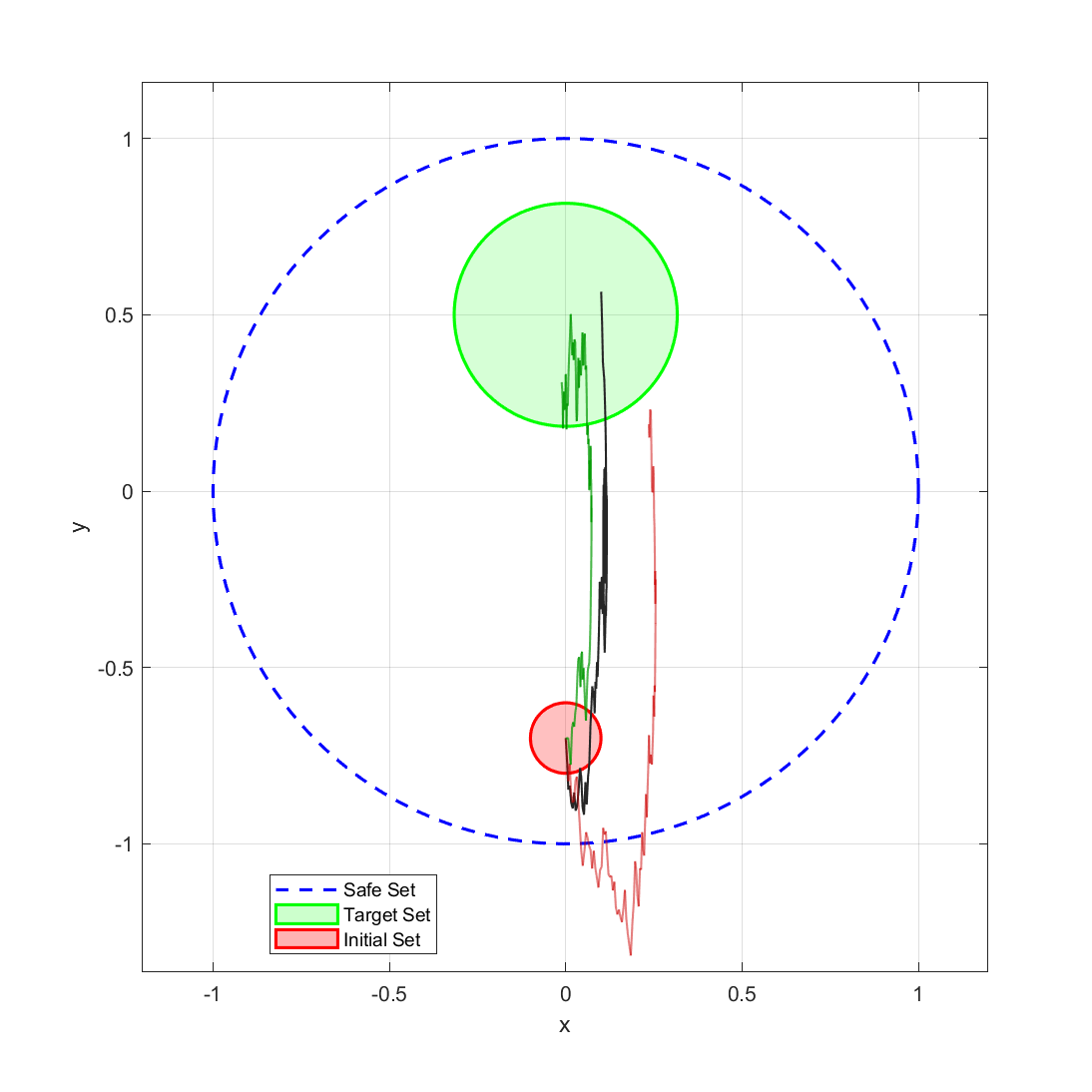}
\caption{An illustration of three system trajectories.}
\label{fig:ex1_2}
\end{figure}

\begin{table}
\caption{\centering Feasibility of SDP \eqref{sos:1}--\eqref{sos:4} for Example \ref{ex1}\\
(\ding{52}: feasible; \ding{55}: infeasible)}
\centering
  \resizebox{0.45\textwidth}{!}{
\begin{tabular}{|c|c|c|c|c|c|}
\hline
Degree&$\gamma$ & $\epsilon=0.70$ & $\epsilon=0.75$ & $\epsilon=0.80$  &$\epsilon=0.83$ 
\\
\hline
8 &0.9  & \ding{55} & \ding{55} & \ding{55} & \ding{55} \\
8 &0.999  & \ding{55} & \ding{55} & \ding{55} & \ding{55} \\
8 &0.9999  & \ding{55} & \ding{55} & \ding{55} & \ding{55} \\
10 &0.9  & \ding{55} & \ding{55} & \ding{55}& \ding{55} \\
10 &0.999  & \ding{52} & \ding{55} & \ding{55}& \ding{55} \\
10 &0.9999  & \ding{52} & \ding{52} & \ding{52}& \ding{55} \\
12 &0.9  & \ding{55} & \ding{55} & \ding{55} & \ding{55} \\
12 &0.999  & \ding{52} & \ding{55} & \ding{55} & \ding{55} \\
12 &0.9999  & \ding{52} & \ding{52} & \ding{52} & \ding{55} \\
14 &0.9  & \ding{55} & \ding{55} & \ding{55} & \ding{55}\\
14 &0.999  & \ding{52} & \ding{55} & \ding{55} & \ding{55}\\
14 &0.9999  & \ding{52} & \ding{52} & \ding{52} & \ding{55}\\
16 &0.9  & \ding{55} & \ding{55} & \ding{55} & \ding{55} \\
16 &0.999  & \ding{52} & \ding{52} & \ding{55} & \ding{55} \\
16 &0.9999  & \ding{52} & \ding{52} & \ding{52} & \ding{52} \\
18 &0.9 & \ding{55} & \ding{55} & \ding{55} & \ding{55} \\
18 &0.999 & \ding{52} & \ding{52} & \ding{55} & \ding{55} \\
18 &0.9999 & \ding{52} & \ding{52} & \ding{52} & \ding{52} \\
20 &0.9 & \ding{55} & \ding{55} & \ding{55} & \ding{55}   \\
20 &0.999 & \ding{52} & \ding{52} & \ding{52} & \ding{52}   \\
\hline
\textbf{MC} & \multicolumn{5}{c|}{\textbf{0.8445}} \\
\hline
\end{tabular}
}
\label{tab:sdp_ex1_set}
\end{table}

\end{example}

\begin{example}
\label{ex2}
We consider the following system, adapted from Example 2 in \cite{xue2021reach},
\begin{equation*}
\begin{cases}
x(l+1)=x(l)+x(l)y(l)+\theta_1(l),\\
y(l+1)=-0.5y(l)+x(l)y(l)+\theta_2(l),
\end{cases}
\end{equation*}
where $\theta_1(l)\in\Theta_1=[-0.01,0.01]$ and $\theta_2(l)\in\Theta_2=[-0.5,0.5]$ are uniform, $\mathcal{X}_0=\{\bm{x}\mid g_{0,1}(x,y) \geq 0\}$ with $g_{0,1}(\bm{x})=-(x+0.5)^2-(y+0.5)^2+0.01$, $\mathcal{X}=\{\bm{x}\mid h(\bm{x})> 0\}$ with $h(\bm{x})=-x^2-y^2+1$,  $\mathcal{X}_r=\{\bm{x}\mid g_{r,1}(\bm{x})>0\}$ with $g_{r,1}(\bm{x})=-x^2-y^2+0.1$, and $\widehat{\mathcal{X}}=\{\bm{x}\mid \widehat{h}(\bm{x})\geq 0\}$ with $\widehat{h}(\bm{x})=-x^2-y^2+3$. The Monte Carlo (\textbf{MC}) simulations yield an empirical worst-case reach-avoid probability of $0.5815$ over the sampled initial states. Three system trajectories, originating in the initial set $\mathcal{X}_0$, are visualized on Fig. \ref{fig:ex2_1}. Similarly, this example satisfies Assumptions \ref{ass}--\ref{ass:poly}.

The verification results via solving SOS program \eqref{sos:1}--\eqref{sos:4} with $g_{\mathrm{b}}=-x^2-y^2+3$ are shown in Table \ref{tab:sdp_ex2_set}.

\begin{figure}[h]
\centering
\includegraphics[width=0.4\textwidth, height=5cm]{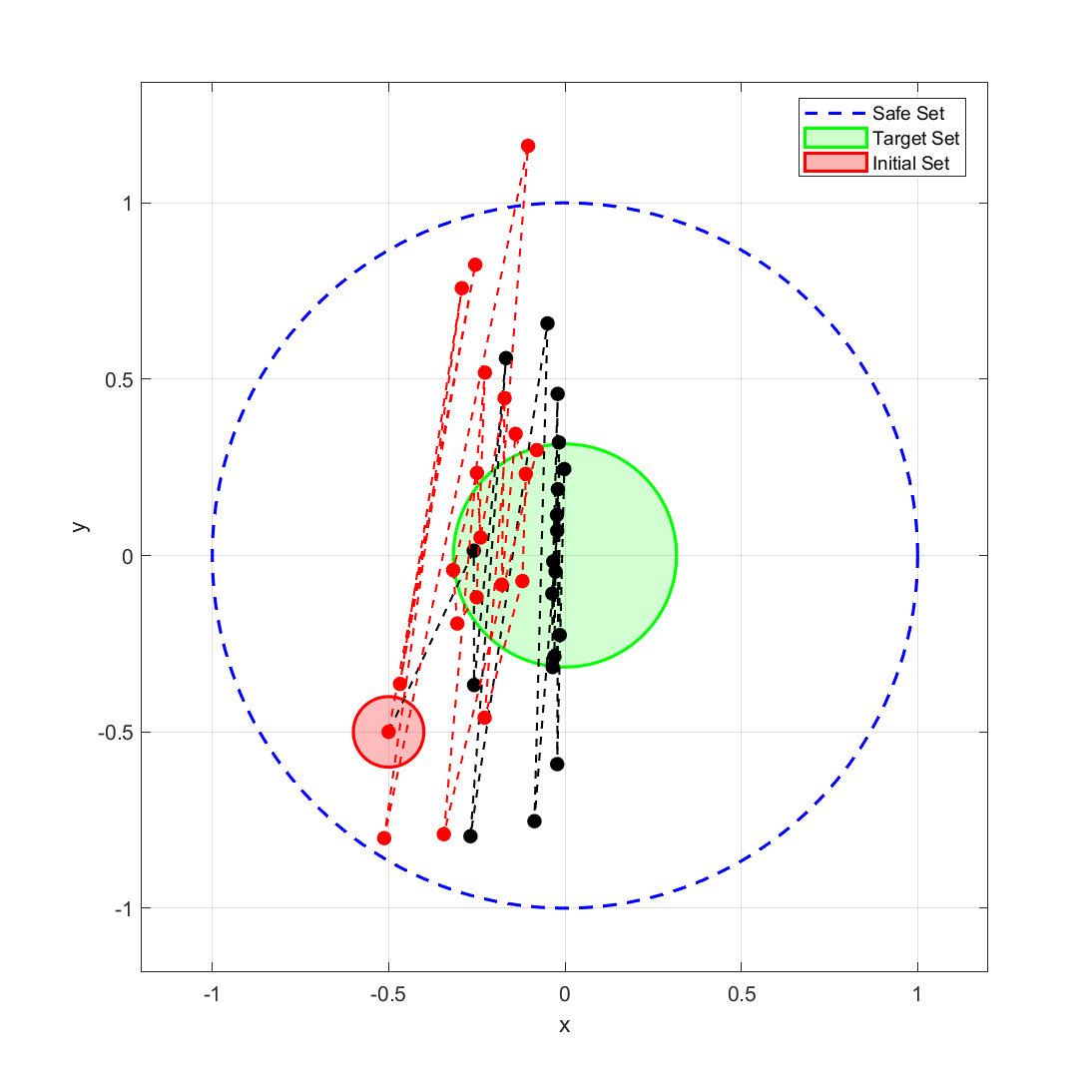}
\caption{An illustration of two system trajectories.}
\label{fig:ex2_1}
\end{figure}

\begin{table}
\caption{\centering Feasibility of SDP \eqref{sos:1}--\eqref{sos:4} for Example \ref{ex2}\\
(\ding{52}: feasible; \ding{55}: infeasible)}
\centering
 \resizebox{0.45\textwidth}{!}{
\begin{tabular}{|c|c|c|c|c|c|}
\hline
Degree&$\gamma$ & $\epsilon=0.30$ & $\epsilon=0.40$ & $\epsilon=0.50$  &$\epsilon=0.55$ 
\\
\hline
8 &0.9  & \ding{55} & \ding{55} & \ding{55} & \ding{55} \\
8 &0.99  & \ding{52} & \ding{52} & \ding{55} & \ding{55} \\ 
8 &0.9999  & \ding{52} & \ding{52} & \ding{55} & \ding{55} \\ 
10 &0.9  & \ding{55} & \ding{55} & \ding{55}& \ding{55} \\
10 &0.999  & \ding{52} & \ding{52} & \ding{55}& \ding{55} \\
10 &0.9999  & \ding{52} & \ding{52} & \ding{55}& \ding{55} \\
12 &0.9  & \ding{55} & \ding{55} & \ding{55} & \ding{55} \\
12 &0.999  & \ding{52} & \ding{52} & \ding{55} & \ding{55} \\
12 &0.9999  & \ding{52} & \ding{52} & \ding{55} & \ding{55} \\
14 &0.9  & \ding{52} & \ding{55} & \ding{55} & \ding{55}\\
14 &0.999  & \ding{52} & \ding{52} & \ding{55} & \ding{55}\\
14 &0.9999  & \ding{52} & \ding{52} & \ding{55} & \ding{55}\\
16 &0.9  & \ding{52} & \ding{55} & \ding{55} & \ding{55} \\
16 &0.999  & \ding{52} & \ding{52} & \ding{55} & \ding{55} \\
16 &0.9999  & \ding{52} & \ding{52} & \ding{52} & \ding{55} \\
18 &0.9 & \ding{52} & \ding{55} & \ding{55} & \ding{55} \\
18 &0.999 & \ding{52} & \ding{52} & \ding{52} & \ding{55} \\
18 &0.9999 & \ding{52} & \ding{52} & \ding{52} & \ding{52} \\
\hline
\textbf{MC} & \multicolumn{5}{c|}{\textbf{0.5815}} \\
\hline
\end{tabular}}
\label{tab:sdp_ex2_set}
\end{table}

\end{example}

The numerical results in Tables \ref{tab:sdp_ex1_set} and \ref{tab:sdp_ex2_set} illustrate several aspects of the theoretical results. First, the feasibility of the SOS conditions is strongly influenced by the choice of the discount parameter $\gamma$. In both examples, taking $\gamma$ closer to $1$ substantially enlarges the range of reach-avoid thresholds that can be certified. For example, in Example \ref{ex1}, no tested threshold is feasible for $\gamma=0.9$, whereas for $\gamma=0.9999$, the threshold $\epsilon=0.80$ is feasible with degree $10$ and $\epsilon=0.83$ is feasible with degree $16$. A similar trend is observed in Example \ref{ex2}: with degree $18$, only $\epsilon=0.30$ is feasible for $\gamma=0.9$, while $\epsilon=0.55$ is feasible for $\gamma=0.9999$. These observations are consistent with the theoretical result that the barrier-like characterization becomes less conservative as $\gamma$ approaches $1$. Thus, in practice, choosing $\gamma$ sufficiently close to $1$ can substantially reduce the conservatism of the resulting convex SDP formulation.

Second, increasing the polynomial degree improves the achievable certification level, particularly when $\gamma$ is sufficiently close to $1$. In Example \ref{ex1}, for $\gamma=0.9999$, the largest certified threshold is $\epsilon=0.80$ for degrees $10$--$14$, and increases to $\epsilon=0.83$ at degree $16$; the same threshold remains feasible at degrees $18$ and $20$. In Example \ref{ex2}, for $\gamma=0.9999$, the largest certified threshold increases from $\epsilon=0.40$ at degrees $8$--$14$ to $\epsilon=0.50$ at degree $16$ and $\epsilon=0.55$ at degree $18$. Thus, higher-degree polynomials can capture increasingly less conservative barrier functions. This behavior is consistent with the completeness result based on polynomial approximation and Putinar's Positivstellensatz: under the stated assumptions and strict feasibility, allowing sufficiently high degrees for the polynomial barrier and SOS multipliers yields a finite level of the hierarchy at which the corresponding polynomial barrier conditions can be certified. At the same time, the tables also show that increasing the degree alone does not compensate for a substantially suboptimal choice of $\gamma$; for example, no tested degree yields a feasible threshold in Example \ref{ex1} for $\gamma=0.9$.

Third, the numerical results show that the certified lower bounds are reasonably close to the empirical reach-avoid probabilities obtained from Monte Carlo simulation. In Example \ref{ex1}, the largest certified threshold is $\epsilon=0.83$, while the empirical worst-case probability over the sampled initial states is $0.8445$, giving a gap of approximately $0.0145$. In Example \ref{ex2}, the largest certified threshold is $\epsilon=0.55$, compared with the empirical value $0.5815$, giving a gap of approximately $0.0315$. These results indicate that the polynomial/SOS certificates provide informative lower bounds for the reach-avoid probability in the two examples.

It is important to emphasize that the Monte Carlo values are empirical estimates and therefore do not constitute formal verification guarantees. In contrast, whenever the SOS program is feasible, the corresponding certified value $\epsilon$ provides, up to numerical errors in solving the SDP, a formal lower bound on the infinite-horizon reach-avoid probability for all initial states in $\mathcal{X}_0$. The Monte Carlo experiments therefore provide an empirical reference for assessing the tightness of the certified lower bounds.  Overall, the numerical experiments illustrate a favorable trade-off between computational tractability and certification tightness: choosing $\gamma$ close to $1$ and allowing higher polynomial and SOS multiplier degrees can lead to stronger certified lower bounds while preserving a sequence of convex SDPs.

%% file: conclusion.tex
\section{Conclusion}
\label{sec:con}
This paper established a sufficient and necessary barrier-like condition involving continuous barrier functions for infinite-horizon reach-avoid verification of stochastic discrete-time systems, under appropriate assumptions, including a uniform absolute continuity condition on the transition kernels. We further established the existence of polynomial barrier functions on compact domains and developed an SOS-based synthesis framework that is complete for the resulting polynomial conditions. Finally, we demonstrated the theoretical results through two numerical examples.

In future work, we will investigate how to relax the assumptions to broaden the applicability of the proposed converse characterization and establish sufficient and necessary barrier-like conditions for  reach-avoid verification of stochastic systems modeled by stochastic differential equations.

%% file: ack.tex
\section*{ACKNOWLEDGMENT}
OpenAI’s GPT-5.6 Luna (free version) \cite{chang2026chatgpt} and Deepseek (free version) \cite{liu2024deepseek} were primarily used to polish the language of this paper prior to submission. No methods or computational results presented in this work were generated by AI models. The authors take full responsibility for the accuracy and integrity of all content reported herein.
